\documentclass[a4paper,conference]{IEEEtran}

\usepackage[utf8]{inputenc}
\usepackage{graphicx}
\usepackage{stfloats}
\usepackage{amssymb}
\usepackage[cmex10]{amsmath}
\usepackage{color}
\usepackage{theorem}
\usepackage{pifont}
\usepackage{array}
\usepackage{cite}
\usepackage{url}
\usepackage{relsize}
\usepackage{tikz}
\usepackage{pgfplots}
\usepackage[english]{babel}
\usepackage{blindtext}
\usepackage{algorithm}
\usepackage{algorithmic}
\usepackage{psfrag}
\usepackage{arydshln} 
\usepackage{paralist}

\def\ve#1{{\mathchoice{\mbox{\boldmath$\displaystyle #1$}}%
              {\mbox{\boldmath$\textstyle #1$}}%
              {\mbox{\boldmath$\scriptstyle #1$}}%
              {\mbox{\boldmath$\scriptscriptstyle #1$}}}}

\newtheorem{myEx}{Example}

\newtheorem{theorem}{Theorem}
\newtheorem{definition}{Definition}

\DeclareMathOperator{\rk}{rk}

\DeclareMathOperator{\dR}{\ensuremath{d_R}}
\DeclareMathOperator{\wtR}{\ensuremath{wt_R}}
\DeclareMathOperator{\In}{\ensuremath{In}}
\DeclareMathOperator{\Out}{\ensuremath{Out}}
\DeclareMathOperator{\PROB}{\ensuremath{\mathrm{Pr}}}

\newcommand{\qed}{\hfill \mbox{\raggedright \rule{.07in}{.1in}}}
\newenvironment{proof}{\vspace{1ex}\noindent{\bf Proof}\hspace{0.5em}}
	{\hfill\qed\vspace{1ex}}

\newcommand{\Fqm}{\ensuremath{\mathbb F_{q^m}}}

\newcommand{\Fq}{\ensuremath{\mathbb F_{q}}}

\newcommand{\Nodes}{\ensuremath{|\mathcal{N}|}}

\newcommand{\Gab}{\mathcal{G}}

\newcommand{\vbeta}{\ensuremath{\boldsymbol{\beta}}}
\newcommand{\PhiB}{\ensuremath{\mathbf{\Phi}_{\boldsymbol{\beta}}}}

\newcommand{\NN}{\ensuremath{\mathbb{N}}}
\newcommand{\PUM}{\ensuremath{\mathcal{PUM}}}
\newcommand{\kone}{\ensuremath{k^{(1)}}}
\newcommand{\Gsub}[1]{\ensuremath{\ve{G}^{(#1)}}}
\newcommand{\C}[1]{\ensuremath{\mathcal{C}_{#1}}}
\newcommand{\csub}[1]{\ensuremath{\ve{c}^{(#1)}}}
\newcommand{\isub}[1]{\ensuremath{\ve{i}^{(#1)}}}
\newcommand{\rsub}[1]{\ensuremath{\ve{r}^{(#1)}}}
\newcommand{\esub}[1]{\ensuremath{\ve{e}^{(#1)}}}

\renewcommand{\a}{\ve a}
\renewcommand{\b}{\ve b}
\renewcommand{\c}{\ve c}
\newcommand{\e}{\ve e}
\renewcommand{\r}{\ve r}

\renewcommand{\i}{\ve i}

\newcommand{\A}{\ve{A}}
\newcommand{\B}{\ve{B}}
\newcommand{\E}{\ve{E}}
\newcommand{\G}{\ve{G}}

\renewcommand{\S}{\ve{S}}
\newcommand{\X}{\ve{X}}
\newcommand{\Y}{\ve{Y}}

\newcommand{\code}[1]{\ensuremath{\mathcal{C}_{#1}}}

\newcommand{\BMD}[1]{\ensuremath{\mathsf{BMD(#1)}}}

\newcommand{\im}{\ensuremath{i \hspace*{0.5mm} \text{--} \hspace*{0.5mm} 1}}

\newcommand{\fprob}[1]{f_{#1}}

\IEEEoverridecommandlockouts

\begin{document}

\title{Error Correction for Differential Linear Network Coding in Slowly-Varying Networks}

\author{\IEEEauthorblockN{Sven Puchinger$^{\diamond}$, Michael Cyran$^{\ast}$, Robert F. H. Fischer$^{\diamond}$, Martin Bossert$^{\diamond}$, Johannes B. Huber$^{\ast}$}
\IEEEauthorblockA{$^{\diamond}$Institute of Communications Engineering, Ulm University, Germany}
\texttt{\{sven.puchinger | robert.fischer | martin.bossert\}@uni-ulm.de}
\IEEEauthorblockA{$^{\ast}$Institute for Information Transmission, University of Erlangen-Nuremberg, Germany}
\texttt{\{cyran | huber\}@LNT.de}
\thanks{This work was supported by Deutsche Forschungsgemeinschaft
		(DFG) within the framework COIN under grants BO~867/29-3, FI~982/4-3, and HU~634/11-3.}
}
\maketitle

\begin{abstract}
Differential linear network coding (DLNC) is a precoding scheme for information transmission over random linear networks.
By using differential encoding and decoding, the conventional approach of lifting, required for inherent channel sounding, can be omitted and in turn higher transmission rates are supported.
However, the scheme is sensitive to variations in the network topology.
In this paper, we derive an extended DLNC channel model which includes slow network changes.
Based on this, we propose and analyze a suitable channel coding scheme matched to the situation at hand using rank-metric convolutional codes.
\end{abstract}



\section{Introduction}

\noindent
\emph{Network coding}, introduced 2000 in \cite{ahlswede2000network}, is a promising method for transmitting information over a network.
\cite{li2003linear} proved in 2003 that \emph{linear network coding} (LNC) is a \emph{max-flow} achieving approach for general multisource multicast networks.
In LNC, the packets within the networks are vectors over a finite field $\Fq$ and intermediate nodes transmit $\Fq$-linear combinations of the incoming packets.
In 2006, \cite{ho2006} showed that choosing the linear functions at the nodes of a network in a random fashion achieves the \emph{max-flow bound} with probability exponentially approaching 1 with the code length.
This method is called \emph{random linear network coding} (RLNC).

In this paper, we consider a unicast scenario, where one source node with $n$ outgoing edges wants to transmit information to one destination node with $\tilde{n} \ge n$ incoming edges over a network whose topology is neither known by the source nor by the destination.
We assume that the transmission takes place in \emph{generations}, which happen infinitely fast, or equivalenty, every node of the network waits until all incoming edges have sent their packets before sending the outgoing packets.
The source emits $n$ vectors from $\Fq^M$ (length-$M$ packets) into the network in parallel and the destination collects $\tilde{n} \geq n$ such packets from the network during one generation.
Within the network, nodes send---possibly different---$\Fq$-linear combinations of the packets of their incoming edges to the outgoing edges.
The vectors sent by the source during one generation can be seen as a matrix $\X \in \Fq^{n \times M}$ and the collected packets by the receiver analogously as $\Y \in \Fq^{\tilde{n} \times M}$, where the rows of the matrices correspond to the packets.
Errors are considered to be vectors added to nodes as additional incoming edges.
This channel model is equivalent to the operator channel described in \cite{koetter2008coding} under the assumption that the network stays constant during one generation.
The input-output relation of the considered scenario is referred to as the \emph{multiplicative additive matrix channel} (MAMC) \cite{silva2010}
\begin{align} \label{Eq_MAMC}
\Y = \A \cdot \X + \B,
\end{align}
where $\A \in \Fq^{\tilde{n} \times n}$ is called the \emph{network channel matrix}, and $\B \in \Fq^{\tilde{n} \times M}$ is the \emph{additive error matrix}. It is commonly assumed (e.g., \cite{silva2010}) that $\tilde{n}=n$ and $\A$ is uniformly distributed at random among all regular matrices from $\Fq^{n \times n}$.
The latter is a good assumption for sufficiently large networks and field sizes.
We also assume that errors occur additively at arbitrary intermediate nodes.
In this case, it can be shown (cf.\ \cite{koetter2008coding}) that the rank of $\B$ is upper bounded by the number of additive errors within the network during one generation.

Error correction for MAMCs can be done using \emph{lifted rank-metric codes}.
Alternatively, \cite{seidl2013differential} showed that higher rates can be achieved by using differential precoding instead of lifting.
This noncoherent transmission scheme for MAMCs is called \emph{differential linear network coding} (DLNC).
However, the method requires the channel matrix to remain constant between generations.
This is a strong assumption which does not hold for all real-world communication networks.
Hence, it has to be determined which influence a ``slow'' variation of the network (change of the multiplicative matrix in the MAMC) has on the error structure.
We show that slow (i.e., the probability of a leaving/joining node during one generation is small) changes of the network topology result in impulsive error peaks.
For this case, we propose a suitable channel coding scheme based on \emph{rank-metric partial-unit-memory (PUM) codes} and analyze it in terms of transmission rates and error probabilities.

This paper is organized as follows: In Sec.~\ref{sec:DLNC} and \ref{sec:RankMetricCodes} we briefly explain DLNC and rank-metric codes.
An error model which includes slow changes of the network topology is introduced in Sec.~\ref{sec:Error_Model} and Sec.~\ref{sec:CodingScheme} shows a proper coding scheme for it and analyzes its performance.
We conclude the paper with a short summary in Sec.~\ref{sec:Conclusion}.

We make use of the following notations.
Let $q$ be a prime power, $m$ a positive integer.
We write $\Fq^{k \times \ell}$ for the set of all $k \times \ell$ matrices over $\Fq$.
Vectors are considered to be row vectors. 
$\ve{I}_n$ denotes the identity matrix of size $n \times n$.
$\mathrm{GL}(n, \Fq)$ is the set of all invertible $(n \times n)$-matrices over $\Fq$.
The entry of a matrix $\A$ in the $i^\mathrm{th}$ row and the $j^\mathrm{th}$ column is called $A_{ij}$.
Let $\vbeta := \left( \beta_1, \dots, \beta_m \right)$ be a basis\footnote{E.g., $\ve{\beta}=\left(1,\alpha, \alpha^2, \dots, \alpha^{m-1}\right)$, where $\alpha$ is a primitive element of $\Fqm$.} of $\Fqm$ over $\Fq$. Then there is a bijective linear map
\begin{align*}
\PhiB : \, \Fqm^n \rightarrow \Fq^{m \times n}, \, \ve{a} = \left(a_1, \dots, a_n \right) \mapsto \A = \left[ A_{ij} \right], \label{eq:bijective_linear_map}
\end{align*}
where $a_j = \sum_{i=1}^{m} A_{ij} \beta_i$ for all $j \in \left\{ 1, \dots, n \right\}$.
$f^{(*)\ell}(x)$ denotes the $\ell$-fold convolution of $f(x)$, i.e., 
$f^{(*)0}(x) = \delta(x)$, $f^{(*)1}(x) = f(x)$, 
$f^{(*)2}(x) = f(x) * f(x)$,
$f^{(*)3}(x) = f(x) * f(x) * f(x)$, etc.

\section{Differential Linear Network Coding}\label{sec:DLNC}

\noindent
In this section, we briefly describe the DLNC precoding method introduced in \cite{seidl2013differential}.
The concept of DLNC can be compared with \emph{differential phase-shift keying}---information is not transmitted absolutely, but by the transition between two successive symbols.

We first restrict ourselves to square matrices, i.e., $n = M$. 
For differential modulation, we assume to have a sequence of source words $\S_i \in \mathrm{GL}(n, \Fq)$, $i=1,\dots,N$, for some $N \in \NN$. Then we generate the DLNC transmit symbol $\X_i$ of the $i$th generation as follows
\begin{equation} \label{Eq_DLNC_X}
 \ve{X}_i = \ve{X}_{\im} \cdot \ve{S}_i,
\end{equation}
where $\X_0 := \ve{I}_n$ is the initialization word.

Differential demodulation at the destination node starts by calculating the pseudoinverse $\ve{Y}_{\im}^+$ of the previously received matrix $\ve{Y}_{\im}^{}$, followed by calculating the product $\ve{Y}_{\im}^+ \cdot \ve{Y}_i^{}$. The pseudoinverse $\ve{Y}_{\im}^+$ has the following properties
\begin{eqnarray}
 \ve{Y}_{\im}^+ \cdot \ve{Y}_{\im}^{} &=& 
  \ve{I}_n^{} + \ve{L} \ve{I}_\mathcal{U}^\mathrm{T}, \label{Eq_pinv1} \\
 \ve{Y}_{\im}^{} \cdot \ve{Y}_{\im}^+ \cdot \ve{Y}_{\im}^{} &=& 
  \ve{Y}_{\im}^{}.
\end{eqnarray}
Thereby, $\ve{I}_\mathcal{U}$ consists of a subset $\mathcal{U}$ of the rows of $\ve{I}_n$. $\ve{L}$ is a full-rank matrix of appropriate dimensions and $\ve{I}_\mathcal{U}^\mathrm{T} \ve{L} = -\ve{I}_{|\mathcal{U}|}$.
Thus, with $\ve{Y}_i = \ve{A}_i \ve{X}_i + \ve{B}_i$, and if we assume that the network channel matrix $\ve{A}_i$ stays constant between the two generations $\im$ and $i$, the demodulation process results in the demodulated symbol $\tilde{\ve{S}}_i$
\begin{eqnarray}
 \tilde{\ve{S}_i} \hspace*{-2mm} &=& \hspace*{-2mm} \ve{Y}_{\im}^+ \ve{Y}_i^{} 
 = \ve{S}_i + 
     \underbrace{%
      \ve{L} \ve{I}_\mathcal{U}^\mathrm{T} \ve{S}_i^{}
      - \ve{Y}_{\im}^+ \ve{B}_{\im}^{} \ve{S}_i^{} 
      + \ve{Y}_{\im}^+ \ve{B}_i^{}
     }_{=: \ve{E}_i} \label{Eq_E_i} \nonumber \\[-0.3cm]
 &=& \hspace*{-2mm} \ve{S}_i + \ve{E}_i, \label{Eq_DLNC_AMC}
\end{eqnarray}
where $\ve{E}_i$ denotes the \emph{effective error matrix} in generation $i$. It can be directly seen from \eqref{Eq_DLNC_AMC} that differential modulation and demodulation transforms the MAMC with network channel matrix $\ve{A}_i$ and additive error matrix $\ve{B}_i$ into an \emph{additive matrix channel} (AMC) \cite{silva2010} with additive error matrix $\ve{E}_i$.
It is shown in \cite{seidl2013differential} that $\rk(\ve{E}_i) \leq \rk(\ve{B}_i) + \rk(\ve{B}_{\im})$.

In case of non-square matrices $\X_i$ ($n \neq M$), the differential encoding reads $\ve{X}_i = \left[\ve{X}_{\im}\right]_{[n]} \cdot \ve{S}_i$, where $\left[\ve{X}_{\im}\right]_{[n]} \in \mathrm{GL}(n,\Fq)$ denotes the square matrix obtained from the first $n$ columns of $\ve{X}_{\im}$.
The restriction to source matrices $\S_i$ such that $\left[\ve{X}_{i}\right]_{[n]}$ is invertible results in a rate loss of $L_{\mathrm{dlnc}} = \frac{n}{q M}$ (cf.\ \cite{seidl2013differential}), which is negligible for sufficiently large field sizes $q$.
Accordingly, differential demodulation is done via $\left[\ve{Y}_{\im}\right]_{[n]}^+ \cdot \ve{Y}_i^{}$.


\section{Rank-Metric Partial-Unit-Memory Codes}\label{sec:RankMetricCodes}

\noindent
In this section, we give an overview of \emph{rank-metric} \emph{partial-unit-memory (PUM) codes}.
We start by defining rank-metric block codes in general and Gabidulin codes in particular.
Next, we review how rank-metric PUM codes are constructed in \cite{wachter2011partial} and give an idea how to decode these codes and state a sufficient decoding condition (cf.\ \cite{wachter2014convolutional,wachter2012rank}).

\subsection{Rank-Metric Codes}

\noindent
Error-correcting codes are often assessed by the Hamming distance of their codewords. However, in 1978, Delsarte \cite{delsarte1978bilinear} introduced the so-called rank metric which turns out to be practical for certain channels, like the MAMC.
The \emph{rank weight} of an element $\mathbf{a} \in \Fqm^n$ is defined as $\wtR\left(\a\right) := \rk\left( \A \right) = \rk\left( \PhiB(\a) \right)$.
Using this notation, we can define the \emph{rank metric} of two elements $\mathbf{a}, \mathbf{b} \in \Fqm^n$ as $\dR\left( \a, \b \right) := \wtR\left( \a - \b\right) = \rk\left( \PhiB(\a) - \PhiB(\b) \right)$.
The \emph{minimum rank distance} of a code $\code{}$ is defined as $d := \min\left\{\dR\left( \a - \b \right) : \a, \b \in \code{} \land \a \neq \b\right\}$.

As for the Hamming distance, a Singleton-like upper bound on the minimum rank distance can be given, i.e., $d \leq n - k + 1$ (cf.\ \cite{delsarte1978bilinear, gabidulin1985theory, roth1991maximum}). Codes fulfilling this bound with equality are called \emph{maximum rank distance (MRD) codes}.
A special class of MRD codes was introduced by Delsarte \cite{delsarte1978bilinear} and later independently reintroduced by Gabidulin \cite{gabidulin1985theory} and Roth \cite{roth1991maximum} and is usually called \emph{Gabidulin codes}. Their structure and known decoding algorithms have a lot in common with \emph{Reed--Solomon codes} in Hamming metric.
\begin{definition}[Gabidulin Code \cite{gabidulin1985theory}]
Let $g_0, g_1, \dots, g_{n-1} \in \Fqm$ be linearly independent over $\Fq$. Then a $\Gab[n,k] \subset \Fqm^n$ Gabidulin code is a linear code given by the following generator matrix:\vspace{-0.3cm}
\begin{align}
\G_{\Gab} =
\begin{bmatrix}
g_0				& g_1			& \dots 	& g_{n-1}				\\
g_0^{q^1}		& g_1^{q^1}		& \dots 	& g_{n-1}^{q^1}	 		\\
\vdots			& \vdots 		& \ddots 	& \vdots	 			\\
g_0^{q^{k-1}}		& g_1^{q^{k-1}}	& \dots 	& g_{n-1}^{q^{k-1}}	 	
\end{bmatrix}
\end{align}
\end{definition}
A proof that these codes are MRD can be found in \cite{gabidulin1985theory}.

\subsection{Partial-Unit-Memory Codes in Rank Metric}

\noindent
We are considering convolutional codes in rank metric. In general, a convolutional code can be described by a semi-infinite block-Toeplitz generator matrix \cite{johannesson1999fundamentals}.
We only consider terminated generator matrices, which is not a restriction of generality in our case because we are dealing with finite information sequences.
Such a generator matrix of a rate $R = \frac{k}{n}$ convolutional code of memory $\mu$ is given by
\begin{align} 
\G = 
\arraycolsep0.1cm
\begin{bmatrix}
\G^{(0)}	& \G^{(1)}	& \dots 	& \G^{(\mu)} 	&\				&				& \ve{0}				\\[0.1cm]
			& \G^{(0)}	& \G^{(1)}	& \dots 		& \G^{(\mu)} 	&				&				\\[0.1cm]
			&			& \ddots	& \ddots		& \ddots		& \ddots		&				\\[0.1cm]
\ve{0}			&			&			& \G^{(0)}		& \G^{(1)}		& \dots 		& \G^{(\mu)}	
\end{bmatrix}, \label{eq:generator_matrix}
\end{align}
where $\G^{(i)} \in \Fq^{k \times n}$ for all $i$ (cf.\ \cite{johannesson1999fundamentals}).

PUM codes are convolutional codes of memory $\mu = 1$, introduced by Lee \cite{lee1976short} and Lauer \cite{lauer1979some}. It can be shown (e.g.\ \cite[Thm. 8.28]{Bossert1999}) that any convolutional code can be represented as a PUM code. However, PUM codes are usually constructed using block codes to obtain a good algebraic understanding of the convolutional code. For the following definition, let $\kone, k \in \NN$, such that $\kone \leq k \leq n-\kone$.
\begin{definition}
A $\PUM(n,k,\kone)$ code over $\Fq$ is a rate $\frac{k}{n}$ convolutional code with memory $\mu=1$ and generator matrix submatrices $\Gsub{0}$ and $\Gsub{1}$ having $\rk\big(\Gsub{0}\big) = k$ and $\rk\big(\Gsub{1}\big) = \kone$.
\end{definition}
W.l.o.g. we can assume that only the first $\kone$ rows of $\Gsub{1}$ are nonzero (otherwise we can transform the information sequence such that it has this form) and therefore we can subdivide the matrices as follows:
\begin{align}
\Gsub{0} =
\begin{bmatrix}
\Gsub{00} \\
\Gsub{01}
\end{bmatrix},
\quad
\Gsub{1} =
\begin{bmatrix}
\Gsub{10} \\
\mathbf{0}
\end{bmatrix},
\end{align}
where $\Gsub{00}, \Gsub{10} \in \Fq^{\kone \times n}$ and $\Gsub{01} \in \Fq^{(k-\kone) \times n}$.

It is shown in \cite{wachter2011partial} and \cite{wachter2014convolutional} how PUM codes can be constructed based on block rank-metric codes.
One can construct PUM codes based on Gabidulin codes by choosing $\Gsub{00}$, $\Gsub{10}$ and $\Gsub{01}$ as submatrices of generator matrices of Gabidulin code~\eqref{eq:generator_matrix}.
In particular, the submatrices are chosen such that the codes in Tab.~\ref{table:codes} are Gabidulin codes with the properties given in the table.
We also choose a \emph{bounded minimum distance (BMD)} error-erasure decoder for each of the defined codes, e.g., from \cite{koetter2008coding} or \cite{gabidulin2008error}.
\begin{table}[h]
	\centering
	\caption{Code definitions using sub-matrices of the PUM code generator matrix \cite{wachter2014convolutional}}
	\label{table:codes}
	\scalebox{0.8}{
	\begin{tabular}{c:c:c:c:c}
	\hline
	  Code
	& Generator Matrix
	& Type
	& Minimum Rank Distance
	& Decoder
	\vphantom{\scalebox{1.5}{M}}\\
	\hline
	  $\C{0}$
	& $\Gsub{0} = \begin{bmatrix} \Gsub{00} \\ \Gsub{01} \end{bmatrix}$
	& $\Gab[n,k]$
	& $d_0 = n - k + 1$
	& $\BMD{\C{0}}$
	\vphantom{\scalebox{1.5}{$\begin{bmatrix} \Gsub{00} \\ \Gsub{01} \end{bmatrix}$}} \\
	\hdashline
	  $\C{1}$
	& $\begin{bmatrix} \Gsub{01} \\ \Gsub{10} \end{bmatrix}$
	& $\Gab[n,k]$
	& $d_1 = n - k + 1$
	& $\BMD{\C{1}}$
	\vphantom{\scalebox{1.5}{$\begin{bmatrix} \Gsub{01} \\ \Gsub{10} \end{bmatrix}$}} \\
	\hdashline
	  $\C{01}$
	& $\Gsub{01}$
	& $\Gab[n,k - \kone]$
	& $d_{01} = n - k + \kone + 1$
	& $\BMD{\C{01}}$
	\vphantom{\scalebox{1.5}{$\begin{bmatrix} \Gsub{01} \\ \Gsub{10} \end{bmatrix}$}} \\
	\hdashline
	  $\C{\sigma}$
	& $\Gsub{\sigma} := \begin{bmatrix} \Gsub{00} \\ \Gsub{01} \\ \Gsub{10} \end{bmatrix}$
	& $\Gab[n,k + \kone]$
	& $d_{\sigma} = n - k -\kone + 1$
	& $\BMD{\C{\sigma}}$
	\vphantom{\scalebox{1.5}{$\begin{bmatrix} \Gsub{00} \\ \Gsub{01} \\ \Gsub{10} \end{bmatrix}$}} \\
	\hline
	\end{tabular}
	}
\end{table}

\subsection{BMD Decoding of Partial-Unit-Memory Codes}\label{subsec:BMDDecPUM}

\noindent
We consider codewords $\c$ of PUM codes of length $nN$, obtained from an information sequence $\i$ of length $k(N-1)$ by multiplication with the (terminated) PUM generator matrix~\eqref{eq:generator_matrix}. We can divide the codeword $\c \in \Fqm^{nN}$ into $N$ blocks of length $n$, $\c = \big( \csub{1} \dots \csub{N}\big)$, as well as the information word $\i \in \Fqm^{k(N-1)}$ into $(N-1)$ blocks of length $k$, $\i = \big( \isub{1} \dots \isub{N-1}\big)$. With that notation, we can derive the simple encoding rule:
\begin{align}
\csub{j} = \Gsub{1} \, \isub{j-1} + \Gsub{0} \, \isub{j} \quad \forall j \in \{1,\dots,N\}, \label{eq:PUMencoding}
\end{align}
with $\isub{0} = \isub{N} = \mathbf{0}$. We use the decoder to correct an additive rank error $\e = \left( \esub{1} \dots \esub{N}\right) \in \Fqm^{nN}$, so the received word is $\r = \c + \e = \left( \rsub{1} \dots \rsub{N}\right) \in \Fqm^{nN}$. 

It is known that convolutional codes can be ML decoded using the \emph{Viterbi} algorithm \cite{viterbi1967error}. However, the complexity of this procedure depends strongly on the size of the underlying field, which defines the number of states of the respective trellis.
The necessary field size over which a Gabidulin code must be defined grows exponentially with the codelength $n$.
Therefore, Viterbi's algorithm is not a good choice for decoding PUM codes based on Gabidulin codes.

An alternative is an algorithm introduced by Dettmar and Sorger \cite{dettmar1995bounded} for decoding PUM codes over Hamming metric.
It uses the block decoders of the underlying algebraic codes to find a much smaller subgraph of the trellis which contains the most likely code sequence under a certain condition.
Afterwards, the Viterbi algorithm finds this sequence in the reduced trellis in less time than without the reduction step.

In \cite{wachter2014convolutional}, a generalization of this algorithm to rank metric was proposed. It is shown that the following bound provides a sufficient condition for successful decoding, using the rank weight distribution of the additive error $\e$.
\begin{align}
&\sum\limits_{h=i}^{i+j-1} t^{(h)} < \frac{\delta_j}{2}, \quad
\begin{cases}
\forall \,i \in \{0, \dots, N\} \text{ and}\\
\forall j \in \{1, \dots, N-i+1\},
\end{cases}	\label{eq:BMDcondition}
\end{align}
where $t^{(h)} := \wtR \left( \E^{(h)} \right) = \wtR \left(\PhiB\left(\esub{h}\right) \right)$ and
\begin{align}
\delta_j := 
\begin{cases}
d_{01}, & \text{if } j=1,\\
d_0 + (j-2) d_\sigma + d_1, & \text{else}.
\end{cases}
\end{align}

If the sequence is sent block-by-block using a Gabidulin code with dimension $k$ instead, all error patterns containing at least one block $h$ with $t^{(h)} \geq \frac{d_0}{2}$ will not be decoded correctly. PUM codes can handle such \emph{error peaks} if $t^{(h)} < \frac{d_{01}}{2}$ and if the errors in the surrounding blocks are not too large (cf.~\eqref{eq:BMDcondition}). This property makes them suitable for DLNC in slowly-varying networks because they can handle seldom occurring network changes resulting in impulsive error peaks.
It can be shown that in terms of the \emph{sequence error probability} $\PROB_\mathrm{fail}$, i.e., the probability that the decoder fails for at least one block of the sequence, the PUM decoder is strictly better than the block-by-block decoder.

\section{DLNC Channel Model for Slowly-Varying Networks} \label{sec:Error_Model}

\noindent
We analyze the statistical behavior of the rank of the effective error matrix $\ve{E}$, cf., (\ref{Eq_E_i}).
For that purpose, we need to discuss the statistical behavior of the additive error matrix $\ve{B}$, and the effects of a slowly-varying network topology.
In order to derive an analytic expression for the probability distributions, we make the following assumptions:
\begin{enumerate}
\item \label{itm:p_n} The probability $p_\mathrm{n}$ that an error occurs at a certain node during one generation is constant.
\item \label{itm:p_d} The probability $p_{\Delta \mathrm{n}}$ that a certain node leaves or joins the network during one generation is constant.
\item \label{itm:p_e} The probability $p_\mathrm{e}$ that there is a directed edge from node $i$ to $j$ ($i \neq j$) is constant.
\item \label{itm:N} The number of nodes $|\mathcal{N}|$ is sufficiently large and can be assumed to be constant due to very slow network changes.
\item \label{itm:q} The field size $q$ is sufficiently large such that the probability that two independent errors cancel is close to zero.
\end{enumerate}

\subsection{Additive Error Matrix} \label{sec:AddErrMatrix}

\noindent
It was shown in \cite{seidl2013differential} that the rank of the additive error matrix $\ve{B}$ is approximately distributed binomially with parameters $|\mathcal{N}|$ and $p_\mathrm{n}$ if Assumption \ref{itm:q} holds.
Hence, its \emph{probability mass function} (pmf) with mean value $\mu_\ve{B} = |\mathcal{N}| \cdot p_\mathrm{n}$ is given by
\begin{align} \label{Eq_rankBDistribution0}
f_{\rk(\B)}(\tau) \approx {|\mathcal{N}| \choose \tau} \cdot p_\mathrm{n}^{\tau} \cdot (1-p_\mathrm{n})^{|\mathcal{N}|-\tau}.
\end{align}

\subsection{Slow Changes of the Network Topology} \label{sec:NetChanges}

\noindent
If nodes leave or join the network between generations, the network behavior changes and can be expressed as a difference in the channel matrices
\begin{equation} \label{Eq_Delta_A_i}
\ve{A}_i = \ve{A}_{\im} + \Delta \ve{A}_i,
\end{equation}
where we call $\Delta \ve{A}_i$ the \emph{channel deviation}.
Since the rank of $\Delta \ve{A}_i$ turns out to be important for the effective error matrix using DLNC, we are interested in its distribution.
We start by proving an important theorem.
\begin{theorem}\label{thm:leaving_node}
If exactly one node $\nu$ with $\In\{\nu\}$ incoming and $\Out\{\nu\}$ outgoing edges leaves the network between generation $i-1$ and $i$, the rank of $\Delta \ve{A}_i$ is upper bounded by $\rk \left( \Delta \ve{A}_i \right) \leq \min\left\{n, \Out\{\nu\}, \In\{\nu\}\right\}$.
\end{theorem}
\begin{proof}
Due to the dimensions of $\Delta \ve{A}_i$, its rank is upper bounded by $n$.
Let the outgoing packets be independent random linear combinations $\ve{x}_{\mathrm{o},j}$ of the incoming edges, sent to nodes $\nu_{\mathrm{o},j}$ with $j \in \{1,\dots,\Out\{\nu\}\}$.
The channel deviation can then be interpreted as errors of value $\ve{e}_{\mathrm{o},j} := -\ve{x}_{\mathrm{o},j}$ at each node $\nu_{\mathrm{o},j}$.
By the same argument as for the additive error matrix (cf.\ \cite{koetter2008coding}), the rank of $\Delta \ve{A}_i$ is then upper bounded by the number of these additive errors, namely $\Out\{\nu\}$.
Alternatively, the channel deviation can be seen as additive errors at the origin nodes of the $\In\{\nu\}$ incoming edges. Hence, $\rk \left( \Delta \ve{A}_i \right) \leq \In\{\nu\}$.
The proof is illustrated in Fig.~\ref{fig:channel_deviation_illustration}.
\end{proof}

\begin{figure}[h]
{
\tikzstyle{N} = [draw,circle, minimum width=1.5cm, inner sep=0pt]
\resizebox{0.5\textwidth}{!}{
\begin{tikzpicture}[scale=1.8, every path/.style={>=latex}] 
	\def\x1{1}
	\def\ya{1.5}
	\def\yb{1.5}
	\def\ed{0.7}
	\def\off{4}

	\node[font=\Large] (title1) at (0,1) {$\rk \left( \Delta \ve{A}_i \right) \leq \Out\{\nu\}$};
	\node[font=\Large] (title1) at (\off,1) {$\rk \left( \Delta \ve{A}_i \right) \leq \In\{\nu\}$};
	\draw[thick] (\off/2,1.5) -- (\off/2,-\ya-\yb-0.5);

	\node[N]		(i1) 	at (-\x1,0)  { $\nu_{\mathrm{i},1}$ };
	\node 			(dots1) at (0,0)  { $\dots$ };
	\node[N] 		(i2) 	at (\x1,0)  { $\nu_{\mathrm{i},\In\{\nu\}}$ };
	\node[N] 		(v) 	at (0,-\yb) { $\nu$ };
	\node[N]		(o1) 	at (-\x1,-\ya-\yb)  { $\nu_{\mathrm{o},1}$ };
	\node 			(dots2) at (0,-\ya-\yb)  { $\dots$ };
	\node[N] 		(o2) 	at (\x1,-\ya-\yb)  { $\nu_{\mathrm{o},\Out\{\nu\}}$ };
	\draw[->,>=latex,thick] (i1) -- (v);
	\draw[->,>=latex,thick] (i2) -- (v);
	\draw[->,>=latex,thick] (v)  -- (o1) node[pos=0.1,left] {$\ve{x}_{\mathrm{o},1}$};
	\draw[->,>=latex,thick] (v)  -- (o2) node[pos=0.1,right] {$\ve{x}_{\mathrm{o},\Out\{\nu\}}$};
	
	\node (e1) at (-\ed*\x1/\ya-\x1,\ed-\ya-\yb) {};
	\draw[->, >=latex,red,thick] (e1) -- (o1) node[pos=0,above] {$-\ve{x}_{\mathrm{o},1}$};
	\node (e2) at (\ed*\x1/\ya+\x1,\ed-\ya-\yb) {};
	\draw[->, >=latex,red,thick] (e2) -- (o2) node[pos=0,above] {$-\ve{x}_{\mathrm{o},\Out\{\nu\}}$};

	\node[N]		(i1) 	at (\off-\x1,0)  { $\nu_{\mathrm{i},1}$ };
	\node 			(dots1) at (\off+0,0)  { $\dots$ };
	\node[N] 		(i2) 	at (\off+\x1,0)  { $\nu_{\mathrm{i},\In\{\nu\}}$ };
	\node[N] 		(v) 	at (\off+0,-\ya) { $\nu$ };
	\node[N]		(o1) 	at (\off-\x1,-\yb-\ya)  { $\nu_{\mathrm{o},1}$ };
	\node 			(dots2) at (\off+0,-\yb-\ya)  { $\dots$ };
	\node[N] 		(o2) 	at (\off+\x1,-\yb-\ya)  { $\nu_{\mathrm{o},\Out\{\nu\}}$ };
	\draw[->,>=latex,thick] (i1) -- (v) node[pos=0.8,left] {$\ve{x}_{\mathrm{i},1}$};
	\draw[->,>=latex,thick] (i2) -- (v) node[pos=0.8,right] {$\ve{x}_{\mathrm{i},\In\{\nu\}}$};
	\draw[->,>=latex,thick] (v)  -- (o1);
	\draw[->,>=latex,thick] (v)  -- (o2);
	
	\node (e1) at (\off-\ed*\x1/\ya-\x1,-\ed) {};
	\draw[->, >=latex,red,thick] (e1) -- (i1) node[pos=0,below] {$-\ve{x}_{\mathrm{i},1}$};
	\node (e2) at (\off+\ed*\x1/\ya+\x1,-\ed) {};
	\draw[->, >=latex,red,thick] (e2) -- (i2) node[pos=0,below] {$-\ve{x}_{\mathrm{i},\In\{\nu\}}$};

\end{tikzpicture}
}
}
\caption{Illustration of the proof of Theorem \ref{thm:leaving_node}.}
\label{fig:channel_deviation_illustration}
\end{figure}
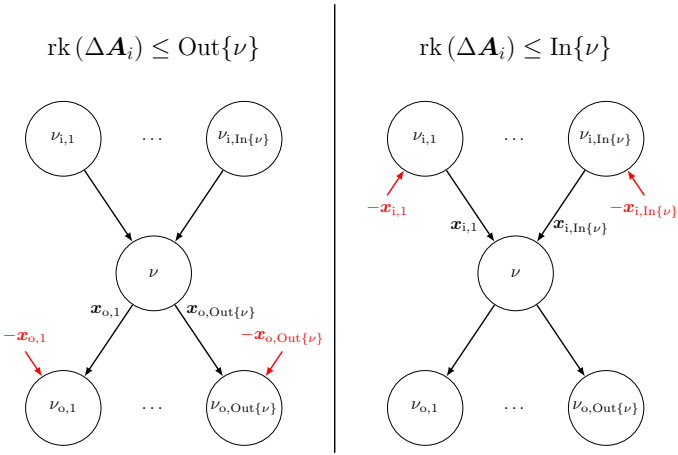

The same argument holds for nodes that join the network between generations.
For a leaving or joining node $\nu$, we define $w(\nu) := \min\left\{n, \Out\{\nu\}, \In\{\nu\}\right\}$, and $\ell_i$ to be the number of leaving nodes in generation $i$.
Its pmf $f_W(w)$ can be derived using Assumption \ref{itm:p_e}.
Due to the subadditivity of the rank, we obtain the following upper bound on $\rk \left( \Delta \ve{A}_i \right)$.
\begin{align*}
\rk \left( \Delta \ve{A}_i \right) \leq \sum\limits_{\substack{\nu \text{ leaving}\\\text{in gen. } i}} w(\nu) = \sum\limits_{j=1}^{\ell_i} w(\nu_j)
\end{align*}
This means that the rank of the deviation matrix $\Delta \ve{A}_i$ is composed by two random processes. The first one determines the number $\ell_i$ of leaving/joining nodes before the current generation according to the pmf $f_L(\ell)$, which is binomially distributed with $|\mathcal{N}|$ and $p_{\Delta \mathrm{n}}$.
The second process determines the $\ell_i$ corresponding node weights $w(\nu_j)$, which are distributed according to $f_W(w)$. 
Hence, the pmf $f_{\rk(\Delta \ve{A}_i)}$ is approximately given by 
\begin{equation} \label{Eq_tauDeltaA}
f_{\rk(\Delta \ve{A}_i)}(\tau) = \sum\nolimits_{\ell=0}^{|\mathcal{N}|} f_L(\ell) \cdot f_W^{(*)\ell}(\tau).
\end{equation}

\subsection{Effective Error Matrix}

\noindent
Based on the insights of the last two paragraphs, we are able to bound the rank of the effective error matrix $\ve{E}_i$.
\begin{theorem}
 The rank of the effective error matrix $\ve{E}_i$, which is present in a 
 DLNC system applied to a slowly-varying network is upper bounded by
 \begin{equation} \label{Eq_theoremEi}
  \rk(\ve{E}_i) \le \rk(\ve{B}_{\im}) + \rk(\ve{B}_i) + \rk(\Delta \ve{A}).
 \end{equation}
Given Assumption \ref{itm:q}, the bound is tight with high probability.
\end{theorem}
\begin{proof}
 Inserting (\ref{Eq_MAMC}) and (\ref{Eq_Delta_A_i}) into (\ref{Eq_E_i}) 
 results in
 \begin{eqnarray}
  \tilde{\ve{S}_i}^{} \hspace*{-2mm} &=& \hspace*{-2mm} 
  \ve{Y}_{\im}^+ \left[(\ve{A}_{\im}^{} + \Delta \ve{A}_i^{}) \ve{X}_i^{} 
    + \ve{B}_i^{} \right] \nonumber \\
  &=& \hspace*{-2mm} \ve{Y}_{\im}^+ \ve{A}_{\im}^{} \ve{X}_i^{}
    + \ve{Y}_{\im}^+ \Delta \ve{A}_i^{} \ve{X}_i^{}
    + \ve{Y}_{\im}^+ \ve{B}_i^{}. \label{Eq_proof_Ei_1}
 \end{eqnarray}
With the aid of (\ref{Eq_pinv1}) $\ve{Y}_{\im}^{} = \ve{I}_n + \ve{L} \ve{I}_\mathcal{U}^{\mathrm{T}}$, we can rearrange 
$\ve{Y}_{\im} = \ve{A}_{\im} \ve{X}_{\im} + \ve{B}_{\im}$ in the following way
\begin{equation} \label{Eq_proof_Ei_2}
 \ve{Y}_{\im}^+ \ve{A}_{\im}^{} = \ve{X}_{\im}^{-1} 
   + \ve{L} \ve{I}_{\mathcal{U}}^{\mathrm{T}} \ve{X}_{\im}^{-1}
   - \ve{Y}_{\im}^{+} \ve{B}_{\im}^{} \ve{X}_{\im}^{-1}.
\end{equation}
Combining (\ref{Eq_proof_Ei_1}) and (\ref{Eq_proof_Ei_2}) we obtain
\begin{small}
\begin{eqnarray}
 \tilde{\ve{S}}_i^{} \hspace*{-2mm} &=& \hspace*{-2mm} \left( \ve{X}_{\im}^{-1} 
   + \ve{L} \ve{I}_{\mathcal{U}}^{\mathrm{T}} \ve{X}_{\im}^{-1}
   - \ve{Y}_{\im}^{+} \ve{B}_{\im}^{} \ve{X}_{\im}^{-1} \right)
   \ve{X}_{i}^{} \nonumber \\
  && \hspace*{-2mm} + \ve{Y}_{\im}^{+} \Delta \ve{A}_{i}^{} \ve{X}_{i}^{}
     + \ve{Y}_{\im}^{+} \ve{B}_{i}^{} \nonumber \\
  &=&\hspace*{-2mm} \ve{X}_{\im}^{-1} \ve{X}_{i}^{}
    + \ve{L}_{}^{} \ve{I}_{n}^{\mathrm{T}} \ve{X}_{\im}^{-1} \ve{X}_{i}^{} 
    - \ve{Y}_{\im}^{+} \ve{B}_{\im}^{} \ve{X}_{\im}^{-1} \ve{X}_{i}^{} \nonumber \\
  &&\hspace*{-2mm} + \ve{Y}_{\im}^{+} \Delta \ve{A}_{i}^{} \ve{X}_{i}^{}
     + \ve{Y}_{\im}^{+} \ve{B}_{i}^{} \nonumber \\
  &=&\hspace*{-2mm} \ve{S}_{i}^{} 
    + \underbrace{\ve{L} \ve{I}_{\mathcal{U}}^{\mathrm{T}} \ve{S}_{i}^{}
    - \ve{Y}_{\im}^{+} \ve{B}_{\im}^{} \ve{S}_{i}^{}
    + \ve{Y}_{\im}^{+} \ve{B}_{i}^{} 
    + \ve{Y}_{\im}^{+} \Delta \ve{A}_{i}^{} \ve{X}_{i}^{}}_{=: \ve{E}_i}.\nonumber%
\end{eqnarray}%
\end{small}%
Thus, the effective error matrix consists of four parts.
The first part corresponds to the rank deficiency of $\ve{Y}_{\im}^{}$, which is assumed to be zero. 
This assumption is justifiable as long as $q$ is sufficiently large (cf. \cite{seidl2013differential}).
The second and the third part have the same rank as the preceding and the current additive error matrices, $\ve{B}_{\im}$ and $\ve{B}_i$, respectively.
The rank of the last part equals to $\rk(\Delta \ve{A}_i)$.
Finally, taking the subadditivity of the rank into account, we obtain \eqref{Eq_theoremEi} as an upper bound on the rank of $\ve{E}_i$.
Due to Assumption \ref{itm:q}, with high probability, the sum of the matrices has the same rank as the sum of the ranks.
\end{proof}

As a consequence, the rank $\ve{E}_i$ can approximately be described by the sum of three random variables
\begin{equation}
\rk(\ve{E}_i) = \rk(\ve{B}_{\im}) + \rk(\ve{B}_i) + \rk(\Delta \ve{A}),
\end{equation}
where the first two summands, which describe the effect of the additive error in the MAMC are distributed according to~\eqref{Eq_rankBDistribution0}. 
The third summand describes the effect of the slowly-varying network and its pmf is given by \eqref{Eq_tauDeltaA}. 
The resulting approximate pmf can be described as 
\begin{equation}
f_{\rk(\E_i)}(\tau)  = f_{\rk(\B)}^{(*)2}(\tau) * f_{\rk(\Delta \A)}(\tau) \label{eq:pdf_E}.
\end{equation}

\begin{myEx}
Fig.~\ref{fig:statistics_id10013} illustrates the influence of the additive rank error and the rank error caused by slow network changes.
The plot shows the pmfs of the ranks of the additive error $\B_i$, the channel deviation $\Delta \A_i$, the additive error without the influence of network changes $\rk(\tilde{\E}_i) := \rk(\B_i)+\rk(\B_{i-1})$ and the effective error $\E_i$.
We consider a network with $|\mathcal{N}| = 100$ nodes and probability parameters $p_\mathrm{n} = 0.03$, $p_\mathrm{e} = 0.05$ and $p_{\Delta \mathrm{n}} = 0.01$.
We have chosen $p_{\Delta \mathrm{n}}$ relatively large such that the changes are more visible in the illustration.
Plot (iii) in Fig.~\ref{fig:statistics_id10013} depicts the effect caused by the differential demodulation in case of no network changes (cf. \cite{seidl2013differential}), i.e., the expected value of the rank of the effective error is approximately doubled with respect to the additive error given by the MAMC (plot (i)).
The effect of additional slow network changes can be seen in plot (iv).
In contrast to (iii), both variance and mean are increased, making errors with high rank more likely.
These error peaks can be better handled using PUM codes than using block codes.
\vspace{-0.3cm}
\begin{figure}[h]
%
%
%
\definecolor{mycolor1}{rgb}{0.00000,0.00000,0.56250}%
\begin{tikzpicture}

\draw (6.4,3.85) 	node {(i)};
\draw (6.4,1.575) 		node {(ii)};
\draw (6.4,-0.7) 		node {(iii)};
\draw (6.4,-2.95) 	node {(iv)};

\small
\begin{axis}[%
width=0.39\textwidth,
height=0.1\textwidth, 
area legend,
scale only axis,
xmin=-1,
xmax=31,
xmajorgrids,
ymin=0,
ymax=0.4,
xticklabels={},
ytick={0, 0.2, 0.4},
ylabel={$\fprob{\rk(\Delta \A_i)}(\tau)$},
ymajorgrids,
name=plot2,
]
\addplot[ybar,bar width=0.01\textwidth,draw=black,fill=mycolor1] plot table[row sep=crcr] {%
0	0.36889018600201\\
1	0.0231180459860032\\
2	0.0586944355852956\\
3	0.0911433918438329\\
4	0.0981777700218436\\
5	0.0822248030761639\\
6	0.0616778091358218\\
7	0.0479092074304523\\
8	0.0395346038355769\\
9	0.032288769474248\\
10	0.0253376739819835\\
11	0.0190948423706681\\
12	0.0142877946335622\\
13	0.0104795373719017\\
14	0.0077693599147917\\
15	0.00567790541638334\\
16	0.0041304360907685\\
17	0.00289948911127024\\
18	0.00205113701325663\\
19	0.00142483196033241\\
20	0.00100829818089914\\
21	0.000696589862860375\\
22	0.000485615082349968\\
23	0.000312791474356274\\
24	0.000227342475974991\\
25	0.000147670311287091\\
26	0.000104103357182831\\
27	7.05255113400449e-05\\
28	4.76588779704058e-05\\
29	2.93655712746945e-05\\
30	1.88950602054386e-05\\
31	1.25164730023288e-05\\
32	7.58209422256455e-06\\
33	7.34139281867361e-06\\
34	4.57332667392783e-06\\
35	3.00876754863673e-06\\
36	9.62805615563753e-07\\
37	7.22104211672815e-07\\
38	7.22104211672815e-07\\
39	4.81402807781876e-07\\
40	7.22104211672815e-07\\
41	2.40701403890938e-07\\
42	1.20350701945469e-07\\
43	1.20350701945469e-07\\
44	0\\
45	0\\
46	0\\
47	0\\
48	0\\
49	0\\
50	0\\
};
\addplot [color=black,solid,forget plot]
  table[row sep=crcr]{%
-1	0\\
51	0\\
};
\end{axis}

\begin{axis}[%
width=0.39\textwidth,
height=0.1\textwidth, 
area legend,
scale only axis,
xmin=-1,
xmax=31,
xmajorgrids,
ymin=0,
ymax=0.4,
xticklabels={},
ytick={0, 0.2, 0.4},
ylabel={$\fprob{\rk(\B_i)}(\tau)$},
ymajorgrids,
at=(plot2.above north west),
anchor=below south west,
]
\addplot[ybar,bar width=0.01\textwidth,draw=black,fill=mycolor1] plot table[row sep=crcr] {%
0	0.0460447343559131\\
1	0.146953863558409\\
2	0.227493275404529\\
3	0.229012341964485\\
4	0.169828440074377\\
5	0.100105306864202\\
6	0.0488663565630247\\
7	0.0204779126374255\\
8	0.00758040931273732\\
9	0.00253518753648131\\
10	0.000789861656868114\\
11	0.000226500021061373\\
12	6.57114832622261e-05\\
13	1.51641884451291e-05\\
14	3.36981965447313e-06\\
15	1.20350701945469e-06\\
16	1.20350701945469e-07\\
17	2.40701403890938e-07\\
18	0\\
19	0\\
20	0\\
21	0\\
22	0\\
23	0\\
24	0\\
25	0\\
26	0\\
27	0\\
28	0\\
29	0\\
30	0\\
31	0\\
32	0\\
33	0\\
34	0\\
35	0\\
36	0\\
37	0\\
38	0\\
39	0\\
40	0\\
41	0\\
42	0\\
43	0\\
44	0\\
45	0\\
46	0\\
47	0\\
48	0\\
49	0\\
50	0\\
};
\addplot [color=black,solid,forget plot]
  table[row sep=crcr]{%
-1	0\\
51	0\\
};
\end{axis}

\begin{axis}[%
width=0.39\textwidth,
height=0.1\textwidth, 
area legend,
scale only axis,
xmin=-1,
xmax=31,
xmajorgrids,
ymin=0,
ymax=0.2,
ytick={0, 0.1, 0.2},
xticklabels={},
ylabel={$\fprob{\rk(\tilde{\E}_i)}(\tau)$},
ymajorgrids,
name=plot3,
at=(plot2.below south west),
anchor=above north west,
]
\addplot[ybar,bar width=0.01\textwidth,draw=black,fill=mycolor1] plot table[row sep=crcr] {%
0	0.00204669877410529\\
1	0.0135444890492528\\
2	0.042498534607397\\
3	0.0879037842804545\\
4	0.13497466310707\\
5	0.163434411139366\\
6	0.1635750249697\\
7	0.139743621075078\\
8	0.103406305566257\\
9	0.0680831289315842\\
10	0.0401536608288798\\
11	0.0215898106674687\\
12	0.0106607388624329\\
13	0.00489937883072907\\
14	0.00212124252520825\\
15	0.000851542619683554\\
16	0.000331455657704919\\
17	0.000120350701945469\\
18	4.31052003906731e-05\\
19	1.32631385817456e-05\\
20	2.70175045183706e-06\\
21	1.842102580798e-06\\
22	2.45613677439733e-07\\
23	0\\
24	0\\
25	0\\
26	0\\
27	0\\
28	0\\
29	0\\
30	0\\
31	0\\
32	0\\
33	0\\
34	0\\
35	0\\
36	0\\
37	0\\
38	0\\
39	0\\
40	0\\
41	0\\
42	0\\
43	0\\
44	0\\
45	0\\
46	0\\
47	0\\
48	0\\
49	0\\
50	0\\
};
\addplot [color=black,solid,forget plot]
  table[row sep=crcr]{%
-1	0\\
51	0\\
};
\end{axis}

\begin{axis}[%
width=0.39\textwidth,
height=0.1\textwidth, 
area legend,
scale only axis,
xmin=-1,
xmax=31,
xmajorgrids,
ymin=0,
ymax=0.2,
ytick={0, 0.1, 0.2},
ylabel={$\fprob{\rk(\E_i)}(\tau)$},
ymajorgrids,
at=(plot3.below south west),
anchor=above north west,
xlabel={$\tau$},
]
\addplot[ybar,bar width=0.01\textwidth,draw=black,fill=mycolor1] plot table[row sep=crcr] {%
0	0.00075378837606254\\
1	0.00500990498557695\\
2	0.0161725553978579\\
3	0.0343824762500784\\
4	0.0557657258148203\\
5	0.0738581205223859\\
6	0.0853304897819184\\
7	0.090394921003887\\
8	0.0898161323730002\\
9	0.0857213839495637\\
10	0.0794393229216877\\
11	0.0711941945768746\\
12	0.0620976464192166\\
13	0.0527427126728921\\
14	0.0438309888075075\\
15	0.0355262991459153\\
16	0.0283831165649355\\
17	0.0223183008347549\\
18	0.0173140449637591\\
19	0.0131960860478045\\
20	0.00997326617927907\\
21	0.00744602424526294\\
22	0.00550346567039209\\
23	0.00402167835439819\\
24	0.00289787297327269\\
25	0.00205197946817025\\
26	0.00147576978089663\\
27	0.00104741952744174\\
28	0.000732174372447844\\
29	0.000502525584041693\\
30	0.00034938545615802\\
31	0.000244754029568694\\
32	0.000167262914336458\\
33	0.000109052472783241\\
34	7.79823425871152e-05\\
35	4.82630876169075e-05\\
36	3.61052105836407e-05\\
37	2.07543557436574e-05\\
38	1.48596274851038e-05\\
39	1.14210360009476e-05\\
40	6.26314877471319e-06\\
41	3.80701200031586e-06\\
42	3.68420516159599e-06\\
43	2.33332993567746e-06\\
44	1.22806838719866e-06\\
45	7.36841032319199e-07\\
46	4.91227354879466e-07\\
47	4.91227354879466e-07\\
48	2.45613677439733e-07\\
49	0\\
50	2.45613677439733e-07\\
};
\addplot [color=black,solid,forget plot]
  table[row sep=crcr]{%
-1	0\\
51	0\\
};
\end{axis}
\end{tikzpicture}%
\caption{Influence of slow network changes on the pmf of $\rk(\E_i)$.}
\label{fig:statistics_id10013}
\end{figure}
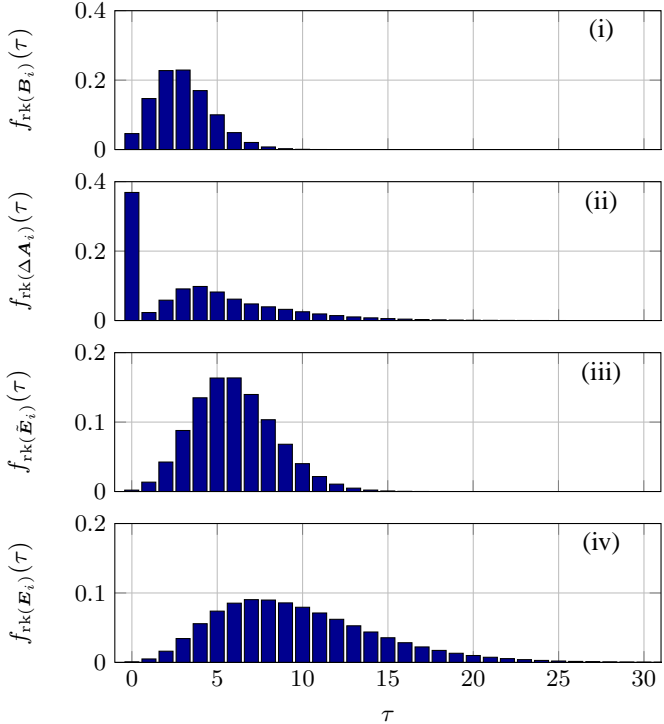
\end{myEx}


\section{Suitable Coding Schemes}\label{sec:CodingScheme}

\noindent
In this section, we describe how PUM codes can be used to make DLNC resilient against the additive error peaks caused by network changes.
We show how PUM codes can be combined with DLNC, give an idea how to choose code parameters and present numerical results which demonstrate that better results can be obtained using PUM codes instead of ordinary block codes.

\subsection{Combining DLNC and PUM Codes}

\noindent
An information sequence $\i^{(j)} \in \Fqm^{k}$ for $j=1,\dots,N-1$ has to be transmitted using PUM codes in combination with DLNC.
We first determine the corresponding codeword sequence $\c^{(j)} \in \Fqm^{n}$ for $j=1,\dots,N$ using the encoding rule \eqref{eq:PUMencoding} and calculate the corresponding matrix representation sequence $\S_j := \PhiB(\c^{(j)})$ for $j=1,\dots,N$, which we use as source symbols for the differential encoding described in Sec.~\ref{sec:DLNC}.
After the transmission, the sequence is demodulated and the result $\tilde{\S}_j = \S_j + \E_j$ ($j=1,\dots,N$) can be decoded using the PUM BMD decoder described in Sec.~\ref{subsec:BMDDecPUM}.
The error sequence $\E_j$ is distributed according to the error model derived in Sec.~\ref{sec:Error_Model}.

\subsection{PUM Code Parameter Choices}

\noindent
When designing a PUM code for DLNC in slowly-varying networks, one has several possibilities to choose the code parameters $N$ and $\kone$.
Here, we assume that $n$ and $k$ are fixed, e.g., because the desired code rate and the packet size is given.
In general, $N$ should be chosen sufficiently large, such that the rate loss of the termination of the PUM code does not play a role, i.e., $\frac{N}{N+1} \cdot \frac{k}{n} \approx \frac{k}{n}$.

It is not easy to analytically derive a good range for $\kone$ for general DLNC channels in slowly-varying networks due to the involved analytic description of the pmf of the effective rank error \eqref{eq:pdf_E}.
However, for a given network, one can use the pmf of $\rk(\E_i)$ to get an idea how to choose $\kone$.
As already mentioned, the PUM decoder described in \cite{wachter2014convolutional} is able to decode up to $\lfloor\frac{d_{01}-1}{2}\rfloor$ errors under certain conditions.
Hence, a necessary condition for good decoding results is that $\PROB\big\{\rk(\E_i) \geq \tfrac{d_{01}}{2} = \tfrac{n-k+\kone+1}{2}\big\}$ is small, or equivalently, that $\kone$ is as large as possible.
On the other hand, $\kone$ should not be chosen too large, because otherwise $d_\sigma = n-k-\kone+1$ gets too small and the decoding capabilities decrease again (cf.\ \eqref{eq:BMDcondition}).
The following figure illustrates this behavior by showing how $\PROB_\mathrm{fail}$ changes as a function of $\kone$ for given $p_\mathrm{n}=0.03$, $p_\mathrm{e}=0.05$, $p_{\Delta \mathrm{n}}=0.005$, $|\mathcal{N}|=100$, $N=50$ and $d_0 = n-k+1 =41$.
It can be seen from Fig.~\ref{fig:k_1choice} that $\kone=24$ is an optimal choice for this given parameter set.
\begin{figure}[h]
%
%
\begin{tikzpicture}
\small
\begin{axis}[%
width=0.4\textwidth,
height=1.5in,
scale only axis,
xmin=0,
xmax=30,
xlabel={$\kone$},
xmajorgrids,
ymode=log,
ymin=0.001,
ymax=1,
yminorticks=true,
ylabel={$\PROB_\mathrm{fail}$},
ymajorgrids,
yminorgrids
]
\addplot [color=blue,solid,mark=asterisk,mark options={solid},forget plot]
  table[row sep=crcr]{%
0	0.224676069870161\\
2	0.151481080480831\\
4	0.101884041718067\\
6	0.0677145485609147\\
8	0.0445697472755982\\
10	0.0298264088420353\\
12	0.0196344226412196\\
14	0.0135380534713882\\
16	0.00957667118693956\\
18	0.00745719631897923\\
20	0.00607214909755127\\
22	0.00538053223852067\\
24	0.00531540179378102\\
26	0.00551832162311634\\
28	0.00615683444263274\\
30	0.0075111632272912\\
};
\end{axis}
\end{tikzpicture}%

\caption{$\PROB_\mathrm{fail}$ as a function of $\kone$.}
\label{fig:k_1choice}
\end{figure}

\subsection{Numerical Results}

\noindent
Example \ref{ex:10012}  shows numerical results obtained by simulating random linear networks with given parameters $\Nodes$, $p_\mathrm{n}$ and $p_\mathrm{e}$ (cf.\ Sec.~\ref{sec:Error_Model}).
In order to run the simulations in sufficiently short time\footnote{Instead of simulating RLNC using real networks over finite fields, we evaluated the statistical behavior of the nodes and counted the number of additive errors and leaving/joining nodes and their number of incoming and outgoing edges.}, we assume that the field size $q$ is large enough, such that independent errors cancel only with negligibly small probability, and therefore the ranks of the error matrices are very likely to be equal to the number of errors happened.
If this assumption does not hold (e.g., if $q$ is relatively small), our results are still upper bounds on the sequence error probability.
We used rank-metric PUM codes with parameters $N$, $n$, $k$ and $\kone$ and checked if the PUM decoder, described in Sec.~\ref{subsec:BMDDecPUM}, is able to correct the error pattern.
$\PROB_\mathrm{fail}$ denotes the probability that the PUM decoder fails, i.e., at least one generation of the sequence is not contained in the subgraph of the trellis.
For comparison, we also checked the cases (block-by-block) when the source symbols for every generation are encoded using Gabidulin block codes with the same code dimension $k$, both differentially and via lifting.

\begin{myEx}\label{ex:10012}
Fig.~\ref{fig:num_res_10012} shows the sequence failure probability $\PROB_\mathrm{fail}$ as a function of the network change probability $p_{\Delta \mathrm{n}}$ of a PUM decoder of a code with $N = 50$, $n = 35$, $k=15$ and $\kone = 10$ and compares it to block-by-block decoders in combination with both DLNC and lifting.
The network is assumed to have $|\mathcal{N}| = 100$ nodes and probability parameters $p_\mathrm{n} = 0.01$ and $p_\mathrm{e} = 0.02$.
It can be seen that the PUM decoder is better than the block-by-block decoder for any $p_{\Delta \mathrm{n}}$ and also improves upon the lifting approach for $p_{\Delta \mathrm{n}} \le 10^{-2}$.

\begin{figure}[h]
%
%
\begin{tikzpicture}
\small
\begin{axis}[%
width=0.4\textwidth,
height=3in,
scale only axis,
xmode=log,
xmin=1e-07,
xmax=0.1,
xminorticks=true,
xlabel={$p_{\Delta n}$},
xmajorgrids,
xminorgrids,
ymode=log,
ymin=1e-05,
ymax=1,
yminorticks=true,
ylabel={$\PROB_\mathrm{fail}$},
ymajorgrids,
yminorgrids,
legend style={at={(0.03,0.97)},anchor=north west,draw=black,fill=white,legend cell align=left}
]

\addplot [color=blue,solid,mark=asterisk,mark options={solid}]
  table[row sep=crcr]{%
0.0316	0.553097345132743\\
0.01	0.0055110386103365\\
0.00316	0.000170892866236799\\
0.001	4.02795129669822e-05\\
0.000316	2.64149898093611e-05\\
0.0001	2.38461613546375e-05\\
3.16e-05	2.34445820954008e-05\\
1e-05	2.25066826842226e-05\\
3.16e-06	2.2570587028166e-05\\
1e-06	2.24247821230594e-05\\
3.16e-07	2.47386093569773e-05\\
};
\addlegendentry{PUM decoder (DLNC)};

\addplot [color=red,solid,mark=asterisk,mark options={solid}]
  table[row sep=crcr]{%
0.0316	0.981563421828909\\
0.01	0.17261307732722\\
0.00316	0.0153671422463229\\
0.001	0.00282055946900861\\
0.000316	0.00107524857517985\\
0.0001	0.000690918679089266\\
3.16e-05	0.000583879502225258\\
1e-05	0.000555404910826123\\
3.16e-06	0.00054556118258681\\
1e-06	0.000541977084204848\\
3.16e-07	0.000541026577845527\\
};
\addlegendentry{block-by-block decoding (DLNC)};

\addplot [color=red,dashed]
  table[row sep=crcr]{%
0.0316	0.0317388354334053\\
0.01	0.0317388354334053\\
0.00316	0.0317388354334053\\
0.001	0.0317388354334053\\
0.000316	0.0317388354334053\\
0.0001	0.0317388354334053\\
3.16e-05	0.0317388354334053\\
1e-05	0.0317388354334053\\
3.16e-06	0.0317388354334053\\
1e-06	0.0317388354334053\\
3.16e-07	0.0317388354334053\\
};
\addlegendentry{block-by-block decoding (lifting)};

\end{axis}
\end{tikzpicture}%
\vspace{-0.2cm}
\caption{$\PROB_\mathrm{fail}$ of different codes/decoders as a function of $p_{\Delta \mathrm{n}}$.}
\label{fig:num_res_10012}
\end{figure}
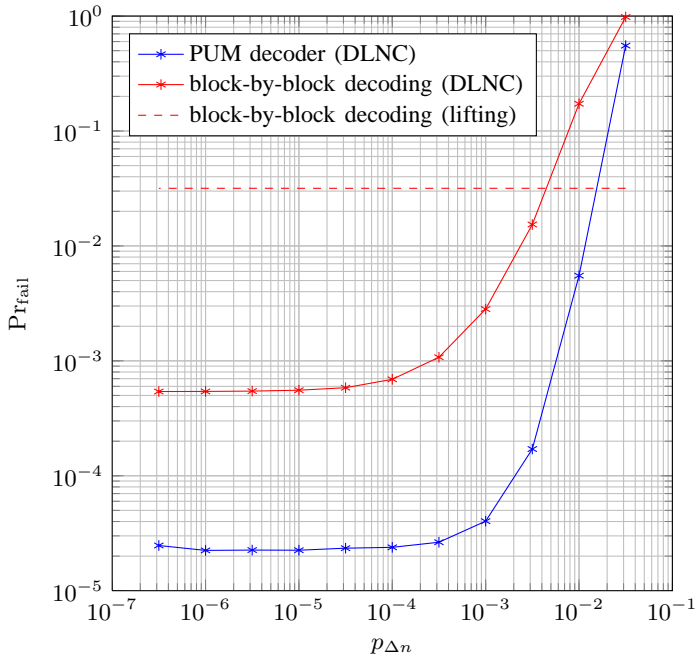

Fig.~\ref{fig:num_res_10012_gain} shows the gain obtained by the new coding scheme, i.e., the fraction of $\PROB_\mathrm{fail}$ of the block-by-block decoder for DLNC and the PUM decoder.
It makes clear that not only the PUM decoder is better at any $p_{\Delta \mathrm{n}}$, but especially good in a region where $p_{\Delta \mathrm{n}}$ is relatively large ($3 \cdot 10^{-4}$ to $10^{-2}$).

\begin{figure}[h]
%
%
\begin{tikzpicture}
\small
\begin{axis}[%
width=0.4\textwidth,
height=1.5in,
scale only axis,
xmode=log,
xmin=1e-07,
xmax=0.1,
xminorticks=true,
xlabel={$p_{\Delta n}$},
xmajorgrids,
xminorgrids,
ymode=log,
ymin=1,
ymax=100,
yminorticks=true,
ylabel={$\PROB_\mathrm{fail}$ gain},
ymajorgrids,
yminorgrids
]
\addplot [color=blue,solid,mark=asterisk,mark options={solid},forget plot]
  table[row sep=crcr]{%
0.0316	1.77466666666667\\
0.01	31.3213333333333\\
0.00316	89.9226666666667\\
0.001	70.0246666666667\\
0.000316	40.706\\
0.0001	28.974\\
3.16e-05	24.9046666666667\\
1e-05	24.6773333333333\\
3.16e-06	24.1713333333333\\
1e-06	24.1686666666667\\
3.16e-07	21.8697247706422\\
};
\end{axis}
\end{tikzpicture}%
\vspace{-0.2cm}
\caption{Gain of $\PROB_\mathrm{fail}$ of PUM codes compared to block-by-block decoding.}
\label{fig:num_res_10012_gain}
\end{figure}
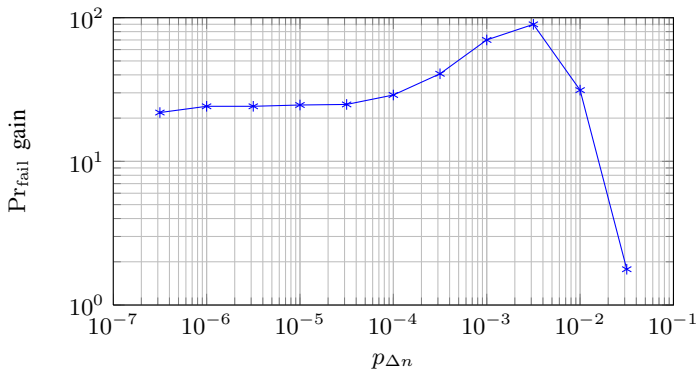
\end{myEx}

\newpage
\section{Conclusion}\label{sec:Conclusion}

\noindent
This paper extended the results of \cite{seidl2013differential} to the case of slowly-varying networks.
At first, we derived a probabilistic DLNC channel model for this case by analyzing the effects of joining/leaving nodes on the network channel matrix, and with that, on the effective error (\ref{Eq_E_i}) in a DLNC system.
Furthermore, we showed that PUM rank-metric codes are the proper error correction strategy for the situation at hand.
We confirmed our considerations by numerical simulations, and showed that in slowly-varying networks, DLNC in combination with rank-metric PUM codes outperforms the conventional, lifting-based RLNC approach.

\bibliographystyle{IEEEtran}
\bibliography{dlnc_pum}

\end{document}